\documentclass[12pt]{article}

\usepackage[colorlinks=true]{hyperref}
\hypersetup{allcolors=[rgb]{0.1,0.1,0.4}}

\usepackage[margin=1.2in]{geometry}
\usepackage{amssymb}
\usepackage{amsmath}
\usepackage{amsfonts}
\usepackage{amsthm}
\usepackage{amscd}
\usepackage{bm}
\usepackage{circuitikz}

\newcommand{\dotarrow}{\Rightarrow}

\usepackage{tikz}
\usepackage{tikz-cd}

\newcommand{\B}{\mathcal{B}}
\newcommand{\C}{\mathcal{C}}
\newcommand{\U}{\mathbb{U}}
\newcommand{\bR}{\mathbb{R}}
\newcommand{\bU}{\mathbb{U}}

\newcommand{\ra}{\rightarrow}
\newcommand{\xra}{\xrightarrow}
\newcommand{\bphi}{\bm{\phi}}
\newcommand{\bpsi}{\bm{\psi}}

\newtheorem{Def}{Definition}[subsection]
\newtheorem{Pro}[Def]{Proposition}

\newtheorem{Cor}[Def]{Corollary}
\newtheorem{Exa}[Def]{Example}

\DeclareMathOperator{\Hom}{Hom}
\DeclareMathOperator{\Obj}{Obj}
\DeclareMathOperator{\id}{id}
\DeclareMathOperator{\im}{im}

\newcommand{\TheTitle}{On the Abstract Structure of\\ the Behavioral Approach to Systems Theory\thanks{The paper is based on Chapter 5 of E. M. Adam's doctoral thesis, see \cite{ADAM:Dissertation}. }
  }

\title{{\TheTitle}}
\author{ Elie M. Adam \thanks{Both authors are with the Laboratory for Information and Decision Systems (LIDS) at the Massachusetts Institute of Technology (MIT).  Emails: \texttt{eadam@mit.edu} and \texttt{dahleh@mit.edu}}
  \and
  Munther A. Dahleh}
\date{}

\begin{document}

\maketitle
\begin{abstract}
  We revisit the behavioral approach to systems theory and make explicit the abstract pattern that governs it.  Our end goal is to use that pattern to understand interaction-related phenomena that emerge when systems interact.  Rather than thinking of a system as a pair $(\U,\B)$, we begin by thinking of it as an injective map $\B\ra\U$.  This relative perspective naturally brings about the sought structure, which we summarize in three points. First, the separation of behavioral equations and behavior is developed through two spaces, one of syntax and another of semantics, linked by an interpretation map. Second, the notion of interconnection and variable sharing is shown to be a construction of the same nature as that of gluing topological spaces or taking amalgamated sums of algebraic objects.  Third, the notion of interconnection instantiates to both the syntax space and the semantics space, and the interpretation map is shown to preserve the interconnection when going from syntax to semantics.  This pattern, in its generality, is made precise by borrowing very basic constructs from the language of categories and functors.

  \vspace{0.1in}
  

\end{abstract}

\section{Introduction}

Our primary concern is studying and understanding phenomena or behavior that arise from the interaction of several systems. We may describe the common situation of interest as small entities of systems coming together, interacting, and producing, as an aggregate, a behavior that would not have occurred without interaction.  These situations are fundamental, and appear in countless settings, including contagion effects in societal systems and cascading failure in infrastructures.  But a viable understanding of the emergent generated behavior ought to be preceded by an understanding of what it means to interconnect systems and have them interact.  The paper begins as an attempt to grasp an understanding of interconnection.

The behavioral approach to systems theory, initiated by J. C. Willems in the 1980's proves to be a pedagogically responsible and natural approach to understanding interconnected systems.  It provides, among other things, a sound inclusive definition of an open system, one that interacts with its environment, and develops interconnection through the natural notion of variable sharing. Rather than viewing a system as an input/output device, the behavioral approach views a system as a set---termed, \emph{behavior}---of trajectories or outcomes deemed allowable by the laws of a mathematical model. Some of the theory's distinctive flair may be sketched through three points.

\begin{itemize}\setlength\itemsep{0em}
 \item[i.] Behaviors are described by equations---termed, \emph{behavioral equations}---and different equations may describe the same behavior.  Intrinsic systemic properties then ought to be properties of the behavior, and not of the descriptive behavioral equations.
 \item[ii.] Many systems are not fundamentally input-output devices, and as such signal-flow diagrams should not be the fundamental interconnection constructs.  Instead, systems are interconnected through the notion of sharing variables.
 \item[iii.] Interconnecting systems on the behavior level via variable sharing coincides with the descriptive interconnection of systems on the equational level.
\end{itemize}

Our interest begins with the following questions. What is the abstract pattern that makes this theory so natural?  How can we mathematically abstract away that pattern, simplify it, and use it in different settings?

Our approach, in answering those questions, consists of developing a relative point of view. Instead of thinking of a system as a pair $(\U,\B)$ with $\B \subseteq \U$, as done in the behavioral approach,  we first explicitly think of a system as an injective map $\B \ra \U$ for arbitrary pairs of $\B$ and $\U$. This view leaves the definition of a system unchanged, but it forces us to introduce transformations, or morphisms, of systems, so as to relate the systems together. The analysis then proceeds naturally, and the findings are summarized again in three points, to reflect those raised above:

\begin{itemize}\setlength\itemsep{0em}
\item[i.] The separation of behavioral equations and behavior suggests a development of two \emph{spaces}, one of syntactical objects (ref. the behavioral equations), and one of semantical objects (ref. the behavior) linked by an interpretation map.
\item[ii.] The morphisms introduce a notion of subsystem and controlled-system. We show that interconnection, as variable sharing, amounts to \emph{gluing} two systems on a common subsystem to yield a controlled-system.  The intuition thus provided by the behavioral approach for interconnection through variable sharing lies on the same level as that of gluing topological spaces or taking amalgamated sums of algebraic objects.
\item[iii.] The notion of interconnection instantiates to both the syntax space and the semantics spaces, and the interpretation map is shown to preserve the interconnection when going from syntax to semantics.
\end{itemize}

The notion of interconnection in the behavioral approach is mathematically clear when the behaviors live in the same universum. It is less clear what that notion amounts to mathematically when the behaviors live in different universa.  A recipe for interconnecting systems in different universa exists, but it does not directly lend itself to mathematical analysis. The key contribution lies in expressing the notion of interconnection through the notion of a pushout.  In a linear setting (e.g., in linear systems), this will allow us to express interconnection of systems in terms of exact sequences from commutative algebra.  Phenomena that emerge from the interaction of systems may then be seen to arise from a certain loss of exactness.  We refer the reader to \cite{ADAM:Dissertation} (e.g., Ch 2) for more details.  The point (ii) further shows a duality between systems with latent variables on one end, and systems that are controlled on another.  Altogether, the abstract notions exhibited by the behavioral approach may be instantiated to different settings.  When those settings are interpreted as defining systems, we recover the intuition provided by the behavioral approach.

We begin by a brief review, in Section 2, of the essential features in the behavioral approach, focusing particularly on three themes: the behavioral equations, variable sharing and latent variables. We then perform, in Section 3, our shift of view to systems being injective maps, and revisit the three themes, highlighting the structure.  The paper will introduce very basic elements of the functorial language (i.e., the language of categories and functors) along the way as needed.  We finally end, in Sections 4 and 5, with a recap of the big picture, and a sketch of where the work leads to.

To dilute the abstraction, we will illustrate the claims by an example. The example runs throughout the paper, and consists, at various stages, of an interplay between two resistive circuits labelled $(S)$ for series and $(P)$ for parallel.
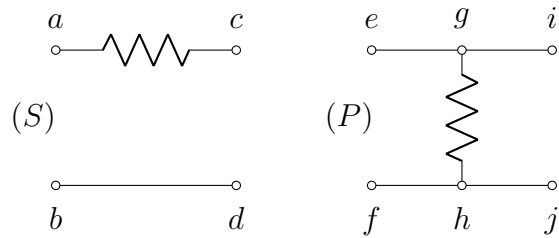
\begin{center}
\begin{circuitikz}[scale=0.6]\draw
  (-0.5,1.5) node[]{$(S)$}
  (6.5,1.5) node[]{$(P)$}
  (0,3) node[label=$a$]{} to[R, o-o] (4,3) node[label=$c$]{}
  (0,0) node[label=below:$b$]{} to [short, o-o] (4,0) node[label=below:$d$]{}
  (7,0) node[label=below:$f$]{} to [open, o-o] (7,3) node[label=$e$]{} -- (11,3) node[label=$i$]{} to[open, o-o] (11,0) node[label=below:$j$]{} -- (7,0)
  (9,3) node[label=$g$]{} to[R, o-o] (9,0) node[label=below:$h$]{};
\end{circuitikz}
\end{center}

In functorial language, the behavioral approach suggests two categories, a syntax category (e.g., reflecting the behavioral equations) and a semantics category (e.g., reflecting the behavior).  The categories are linked via an interpretation functor.  Interconnection (e.g., through variable sharing) in both categories consists of taking pushouts or more generally colimits (or dually pullback/limits in the case of the behavioral approach).  The interpretation functor preserves pushouts or more generally colimits, and may be desired to admit a right adjoint (e.g., reflecting that every behavior admits a \emph{simplest} behavioral equation representation).

\tableofcontents

\section{Jan Willems' behavioral approach}

The behavioral approach to systems theory begins from the premise that a mathematical model acts as an exclusion law.  The phenomenon we wish to model produces events or outcomes that live in a given set $\U$. The laws of the model (viewed descriptively) state that some outcomes in $\U$ are possible, while others are not. The model then restricts the outcomes in $\U$ to only those that are allowed possible by the laws of the model. The set of possible outcomes is then called the \emph{behavior} of a model. We refrain from using the term \emph{model}, and replace it by \emph{system}. Material in this section may be found in \cite{POL1998}, \cite{WIL1991} and \cite{WIL2007}.

\begin{Def}[cf. \cite{POL1998}, Section 1.2.1]
  A Willems system is a pair $(\U,\B)$ where $\U$ is a set, called the universum---its elements are called \emph{outcomes}---and $\B$ a subset of $\U$ called the \emph{behavior}.
\end{Def}

The behavioral approach links naturally to standard ideas.  A dynamical system may be obtained by considering universa of the form $\mathbb{W}^{\mathbb{T}}$, the set of maps from $\mathbb{T}$ to $\mathbb{W}$.  The set $\mathbb{T}$ embodies the time axis, and $\mathbb{W}^{\mathbb{T}}$ then represents timed trajectories taking values in $\mathbb{W}$. An input-output structure can be recovered by thinking of a map as a relation.  Every set map $f: A \ra B$ defines a relation $R = \{(a,fa)\} \subseteq A \times B$ which yields a Willems system $(A\times B, R)$. The universa may also be endowed with additional structure, e.g., a vector space structure.  A Willems system $(\U,\B)$ may then be termed $k$-linear if $\U$ is a vector space over the field $k$, and $\B$ is a linear subspace of $\U$.  Time invariance, among other things, for dynamical systems can be further brought into the picture.  We refer the reader to \cite{POL1998}, \cite{WIL1991} and \cite{WIL2007} for the details.

\begin{Exa}
Considering the circuits (S) and (P),
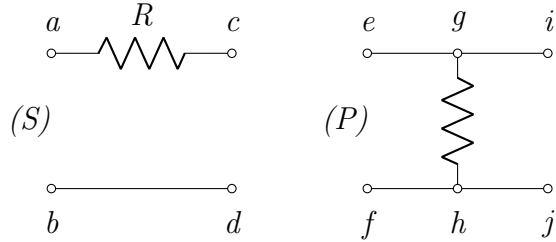
\begin{center}
\begin{circuitikz}[scale=0.6]\draw
  (-0.5,1.5) node[]{(S)}
  (6.5,1.5) node[]{(P)}
  (0,3) node[label=a]{} to[R, l=R ,o-o] (4,3) node[label=c]{}
  (0,0) node[label=below:b]{} to [short, o-o] (4,0) node[label=below:d]{}
  (7,0) node[label=below:f]{} to [open, o-o] (7,3) node[label=e]{} -- (11,3) node[label=i]{} to[open, o-o] (11,0) node[label=below:j]{} -- (7,0)
  (9,3) node[label=g]{} to[R, o-o] (9,0) node[label=below:h]{};
\end{circuitikz}
\end{center}
we declare the variables in play to be the voltage potentials $v_a, \cdots, v_j$, one for each labelled node, and currents $i_{ac}, i_{bd}, i_{eg}, i_{gi}, i_{gh}, i_{fh}, i_{hj}$, one for each consecutive pair of labelled nodes. We define $(\U_S,\B_S)$ and $(\U_P,\B_P)$ to be the Willems systems corresponding to (S) and (P). The universum $\U_S$ is the free $\bR$-vector space (isomorphic to $\bR^6$) generated by the basis $\{v_a,v_b,v_c,v_d,i_{ac},i_{bd}\}$. The behavior $\B_S$ is the subset $\{ (V_a,V_b,V_c,V_d, I_{ac}, I_{bd}) \in \U_S : V_a - V_c = RI_{ac} \text{ and } V_b = V_d\}$.  Similarly, $\U_P$ is the $\bR$-vector space (isomorphic to $\bR^{11}$) generated by the variables that remain. The behavior $\B_P$ is the subset of $\U_P$ that satisfy KCL, KVL and Ohm's law. 
\end{Exa}

\subsection{Behavioral equations}

Systems may be generally described by equations.  The behavior then consists of the outcomes for which \emph{balance equations} are satisfied.  

\begin{Def}[cf. \cite{POL1998}, Section 1.2.2]
  Let $\U$ be a universum, $\mathbb{E}$ a set, and $f_1, f_2 : \U \ra \mathbb{E}$ maps.  The Willems system $(\U,\B)$ with $\B = \{u\in\U : f_1(u) = f_2(u)\}$ is said to be described by \emph{behavioral equations} and is denoted by $(\U,\mathbb{E},f_1,f_2)$.  We call  $(\U,\mathbb{E},f_1,f_2)$ a \emph{behavioral equation representation} of $(\U,\B)$.
\end{Def}

If both $\U$ and $\mathbb{E}$ share a linear structure (e.g., are vector spaces), then the behavior of the representation $(\U,\mathbb{E},f_1,f_2)$ is the kernel of $f_1 - f_2$.  In such a setting, we talk about kernel representations of systems.

Systems are described in many situations using inequalities rather than equalities.  Such a change may be remedied by considering $\mathbb{E}$ to be ordered. For instance, if $\mathbb{E}$ is a partially ordered set where every pair of elements $(a,b)$ admit a least upper bound $max(a,b)$, then $f(u) \leq g(u)$ if, and only if, $max(f(u),g(u)) = g(u)$. 

A system $(\U,\B)$ may have different behavioral equation representations of it.  It is then not the equations themselves that are essential, but rather the solution to those equations.  This remark is the basis for a separation between syntax and semantics.  The behavioral equations represent the syntax, while the semantics, the objects behind the syntax, are captured by the behavior.

\begin{Exa}
The systems $(\U_S,\B_S)$ and $(\U_P,\B_P)$ possess a linear structure. Both $\B_S$ and $\B_p$ are the solution set of a system of linear equations.  We can then explicitely define matrices (or linear maps) with the equations as rows, and obtain behavioral equation representations of the two systems.  
\end{Exa}

\subsection{Interconnection and variable sharing}

The behavioral approach enables us to define interconnections of systems. Let $(\U,\B)$ and $(\U,\B')$ be Willems systems with representations $(\U,\mathbb{E},f,g)$ and $(\U,\mathbb{E'},f',g')$, respectively.  Their interconnection is the system represented by $(\U,\mathbb{E}\times\mathbb{E'},f\times f',g\times g')$.

\begin{Def}
  The interconnection of the systems $(\U,\B)$ and $(\U,\B')$ is the system $(\U,\B\cap\B')$. \hfill\qed
\end{Def}

To interconnect two systems $(\mathbb{V} \times \U, \B)$ and $(\U \times \mathbb{V}', \B')$ that share only a part $\U$ of their universa in common, we first lift them to two equivalent systems $(\mathbb{V} \times \U \times \mathbb{V}', \B\times\mathbb{V}')$ and $(\mathbb{V} \times \U \times \mathbb{V}', \mathbb{V}\times \B')$, and then intersect the lifted behaviors.

\begin{Def}\label{Def:InterconnectDifferent}
  The interconnection of $(\mathbb{V} \times \U,\B)$ and $(\U \times \mathbb{V}',\B')$ by sharing $\U$ is the system $(\mathbb{V}\times\U\times \mathbb{V}',(\B\times\mathbb{V}') \cap (\mathbb{V}\times\B'))$. \hfill \qed
\end{Def}

Variable sharing thus consists of declaring parts of the universa as representing the same outcomes, and carrying out the above procedure of identification. The identification is the basis for the \emph{gluing} mentioned in the introduction.  Definition \ref{Def:InterconnectDifferent} provides a means to interconnect systems in different universa.  However, that means can be mathematically cumbersome.  Part of the relative perspective to be developed goes into making interconnection less cumbersome when different universa are involved.  
\begin{Exa}
  We will interconnect (S) and (P) by connecting terminal $c$ to $e$, and $d$ to $f$ to obtain the circuit:
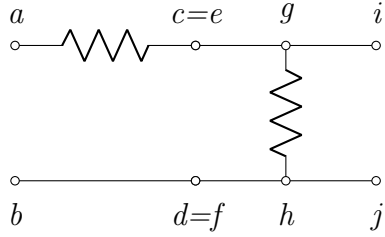
\begin{center}
\begin{circuitikz}[scale=0.6]\draw
  (0,3) node[label=a]{} to[R, o-o] (4,3)
  (0,0) node[label=below:b]{} to [short, o-o] (4,0)
  (4,0) node[label=below:{d=f}]{} to [open, o-o] (4,3) node[label={c=e}]{} -- (8,3) node[label=i]{} to[open, o-o] (8,0) node[label=below:j]{} -- (4,0)
  (6,3) node[label=g]{} to[R, o-o] (6,0) node[label=below:h]{};
\end{circuitikz}
\end{center}  
To perform the interconnection, we need to identify the variables $v_c$, $v_d$, $i_{ac}$ and $i_{bd}$ with $v_e$, $v_f$, $i_{eg}$ and $i_{fh}$, respectively. The systems $(\U_S,\B_S)$ and $(\U_P,\B_P)$ need to be lifted to a common universum $\U$, in accordance with Definition \ref{Def:InterconnectDifferent}, where every pair of the to-be-matched variables corresponds to the same dimension.  The lifted behavior are then intersected. 
\end{Exa}

\subsection{Latent variables}

The universum typically represents the variables that we wish to model.  It is however often the case that auxiliary variables are needed.  Adding auxiliary variables might lead to simpler behavioral equation representations.  Interconnecting two systems that live in different universa will also force us to add auxiliary variables.  Latent variables of auxiliary interest are then appended to the universum of the original manifest variables.

\begin{Def}[cf. \cite{POL1998}, Section 1.2.3]
  A Willems system with latent variables is defined as a triple $(\U,\U_l,\B_f)$ with $\U$ the universum of the manifest variables, $\U_l$ the universum of latent variables and $\B_f \subseteq \U \times \U_l$ the full behavior.  It defines the manifest Willems system $(\U,\B)$ with $\B:= \{u \in \U : (u,l)\in \B_f \text{ for some }l \in \U_l\}$ where $\B$ is the manifest behavior.  We call $(\U,U_l,\B_f)$ a latent variable representation of $(\U,\B)$.
\end{Def}

Latent variables equip us with the extra flexibility needed in the modelling exercise. The theory of latent variables will thereafter appear in the notion of subsystems.

\begin{Exa}
 We may abstract the circuit (P) into a two-port blackbox, by declaring the universum to consist of the voltage potentials and currents at the four terminals $e$, $f$, $i$ and $j$.  The corresponding Willems system $(\U_{two-port},\B_{two-port})$ consists of $\bR^8$ as a universum, and the set of tuples in $\bR^8$ that can physically coincide as the behavior.  We define $U_l$ to be $\bR^3$ generated by the variables $v_g$, $v_h$ and $i_{gh}$.  The Willems system $(\U_P,\B_P)$ is equivalent to $(\U_{two-port},\U_l,\B_P)$ and is then a latent variable representation of $(\U_{two-port},\B_{two-port})$.
\end{Exa}

\subsection{The immediate mathematical structure}
Let $\U$ be a fixed universum, and suppose that every subset of $\U$ is a potential behavior. We can partially order the behaviors in $\U$ by inclusion, and get a lattice $\mathcal{L}_{\U}$.  The meet (min) in the lattice corresponds to set-intersection, and the join (max) corresponds to set-union.  Interconnection of systems corresponds then to taking meets in the lattice. The properties of the lattice (as well as its existence) changes as different mathematical structures are imposed on the universa and the behaviors. For instance, if $U$ is a vector space and the behaviors are the linear subspaces of $\U$, then the lattice of behaviors is modular. For a thorough study along those lines, we invite the reader to look at \cite{SHA2001}.

\section{The relative point of view}

We bring about the abstract structure by adopting a relative point of view. Instead of thinking of a system as a pair of sets $(\U,\B)$ where $\B \subseteq \U$, we explicitly think of it as an injective map $\B \ra \U$ of sets.

{\bf Remark:} We only consider, in this section, universa and behaviors that are sets, without any additional structure.  We can nevertheless equip the systems with more structure (such as an $R$-module structure) while keeping the insight and the result statements unchanged.  We would however need to equip the set maps with a compatible structure. For instance, in the case of $R$-modules, the set maps would have to be replaced by $R$-linear maps.

\subsection{Morphisms of systems}\label{morphism}

A system is then an \emph{injective} map $B \ra U$ of sets. We will keep the labels $B$ and $U$, instead of using other letters, simply to make the connection explicit with the behavioral approach as described.  We will now revisit the above theory through the lens of injective maps.  However, the systems thus far, as simply a collection of injective maps without any further structure, are unrelated.  They cannot be interconnected and we cannot discuss most of the themes addressed in the previous section.  We remedy this issue by defining a morphism of systems.  The systems considered along with their morphisms will provide us with a \emph{sandbox} to develop the theory we want.

\begin{Def}
  Let $s: B \rightarrow U$ and $s': B' \rightarrow U'$ be two systems. A morphism $\bphi$ from $s$ to $s'$ denoted by $\bphi:s \dotarrow s'$ is a pair of set maps $(\phi_B,\phi_U)$ with $\phi_B:B \rightarrow B'$ and $\phi_U:U \rightarrow U'$ such that the diagram:
  \begin{equation*}
  \begin{CD}
    B  @>s>> U\\
    @VV\phi_BV         @VV{\phi_U}V\\
    B'  @>s'>>   U'
  \end{CD}
  \end{equation*}
  commutes, i.e., such that $\phi_Us = s'\phi_B$.
\end{Def}

If $\id_A$ denotes the identity map on the set $A$, then for every system $B \ra U$, the pair $(\id_B,\id_U)$ is a morphism of systems.  Furthermore, morphisms may be composed component-wise to yield other morphisms.  Indeed, if $\bphi=(\phi_B,\phi_U):s \dotarrow s'$ and $\bphi'=(\phi'_B,\phi'_U):s' \dotarrow s''$ are morphisms, then the composition ${\bphi'\bphi}=(\phi'_B\phi_B,\phi'_U\phi_U): s \dotarrow s''$ is also a morphism.

Generally, given a diagram:
\begin{equation*}
  \begin{CD}
    B  @>s>> U\\
    @VV\phi_BV         @VV{\phi_U}V\\
    B'  @>s'>>   U'
  \end{CD}
  \end{equation*}
where $s$ and $s'$ are injective, it follows that $\phi_B$ is just the \emph{restriction} of $\phi_U$ onto $B$.  If either $s$ or $s'$ were not injective, then $\phi_B$ is not necessarily the restriction of $\phi_U$.  We cannot however always construct a commutative diagram by restricting an arbitrary $\phi_U$ onto $B$, as the image $\phi_U(B)$ may fall outside $B'$. We can then do so (if and) only if $\phi_U(B) \subseteq B'$.

The idea of a morphism, in a different form, appears in \cite{FUH2001} and \cite{FUH2002} through the notion of behavior homomorphism.  Behavior homomorphisms were partly introduced to assist in settling problems regarding equivalence of system representations. 

\subsubsection{Subsystems and controlled-systems}\label{subcon}

Introducing morphisms immediately introduces notions of a subsystem and a controlled-system.

Controlling the behavior of a system consists of restricting some of its potential outcomes.  As such if $(\U,\B)$ and $(\U,\B')$ are Willems systems with $\B \subseteq \B'$, then $\B$ is a controlled version of $\B'$.  Such a notion lifts naturally to the relative perspective.

\begin{Def}[Controlled-system]
  Let $\bphi : s_{ctrl} \dotarrow s$ be a morphism of systems. The pair $(s_{ctrl},\bphi)$ is said to be a controlled-system from $s$, if the components of $\bphi$ are injective maps. We may refer to $s_{ctrl}$ as the controlled-system if ${\bphi}$ is clear from the context. 
\end{Def}

\begin{Exa}
  The system underlying circuit (S$_c$) can be seen as a controlled-system from that of (S): 
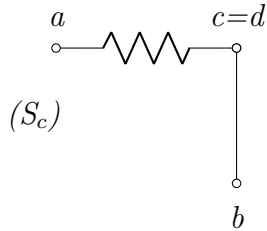
\begin{center}
\begin{circuitikz}[scale=0.6]\draw
  (-0.5,1.5) node[]{(S$_c$)}
  (0,3) node[label=a]{} to[R, o-o] (4,3) node[label={c=d}]{}
  (4,0) node[label=below:b]{} to [short, o-o] (4,3); 
\end{circuitikz}
\end{center}
Let $s_c: \B_{S_c} \ra \U_{S_c}$ and $s: \B_S \ra \U_S$ be the systems of (S$_c$) and (S), respectively.  Then $\U_{S_c}$ is the free $\bR$-vector space with basis $\{v_{a'}, v_{b'}, v_{c'}, i_{a'c'}, i_{b'c'}\}$.  The set $\B_{S_c}$ is the subset of $\U_{S_c}$ whose tuples satisfy the laws of the circuit. The morphism $\bphi: s_c \dotarrow s$ is defined uniquely such that $\phi_U$ sends $v_{a'}, v_{b'}, v_{c'}, i_{a'c'}, i_{b'c'}$ in the basis of $\U_{S_c}$ respectively to $v_{a}, v_{b}, v_{c} + v_{d}, i_{ac}, i_{bd}$. The pair $(s_c,\bphi)$ is then a controlled-system from $s$.
\end{Exa}

Dually, a notion of subsystem, in the behavioral approach, is partially hinted at from the theory of latent variables.  It can be generally thought that a subsystem of a big system consists of a projection of the big system onto only the variables of interest.  We arrive at the following observation:

\begin{Def}[Subsystem]
  Let ${\bphi} : s \dotarrow s_{sub}$ be a morphism of systems. The pair $(s_{sub},\bphi)$ is said to be a subsystem of $s$, if the components of $\bphi$ are surjective maps. We may refer to $s_{sub}$ as the subsystem if ${\bphi}$ is clear from the context.
\end{Def}

\begin{Exa}
  The system underlying the circuit (P$_s$) can be seen as a subsystem of that of (P): 
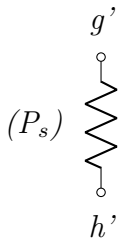
\begin{center}
\begin{circuitikz}[scale=0.6]\draw
  (-1.5,1.5) node[]{(P$_s$)}
  (0,3) node[label=g']{} to[R, o-o] (0,0) node[label=below:h']{};
\end{circuitikz}
\end{center}
Let $p_s: \B_{P_s} \ra \U_{P_s}$ and $p: \B_P \ra \U_P$ be the systems of (P$_s$) and (P), respectively. Then $\U_{P_s}$ is the free $\bR$-vector space with basis $\{v_{g'},v_{h'},i_{g'h'}\}$. The set $\B_{P_s}$ is the subset of $\U_{P_s}$ whose tuples satisfy the laws of the circuit. The morphism $\bphi: p \dotarrow p_s$ is defined uniquely such that $\phi_U$ sends $v_g, v_h, i_{gh}$ in the basis of $\U_P$ respectively to $v_{g'}, v_{h'}, i_{g'h'}$, and everything else remaining in the basis of $\U_P$ to $0$. The pair $(p_s,\bphi)$ is then a subsystem of $p$.
\end{Exa}

The use of these two notions will appear in the interconnection of systems.  Informally, two systems are interconnected by \emph{gluing} them along a common subsystem to yield a controlled-system.  This approach will embody the nature of variable sharing stressed at by the behavioral approach.

{\bf Remark:}  Although we may think of the controlled-systems and the subsystems as the domains or co-domains of the morphisms, the notion however is really embedded in the morphism.  Indeed, two different morphisms from $s$ to $s_{sub}$ with surjective components yield different subsystems.

{\bf Remark:} The prefix sub of subsystem typically alludes to a possibility of embedding.  It seems counterintuitive that surjective maps rather than injective maps are involved.  Similarly, controlled systems would advocate identifying parts of the system together as a means of control.  This hints that surjectve maps rather than injective maps are to be in play. Such an unease will be remedied in a future section, by simply \emph{reversing} the direction of the morphisms.

\subsubsection{Recovering the fixed point of view}

To recover a fixed point of view, we fix the codomains of our systems. We allow only systems of the form $B \ra U$ for a fixed set $U$, and allow only degenerate morphisms that only map $U$ identically to itself.  If $s$ and $s'$ are two systems, we then define a partial-order $s \leq s'$ if, and only if, $s$ factors through $s'$, i.e., $s = s'h$ for some map $h$.  Note that if $s = s'h$ and $s$ is injective, then $h$ is injective. If $s = s'h$ and both $s$ and $s'$ are injective, then $h$ is injective and unique. We then obtain a lattice isomorphic to the lattice of behavior of the universum $U$.  With a fixed universum, the system $B \ra U$ is a controlled-system from $B' \ra U$ if, and only if, $B \subseteq B'$.  The notion of subsystem, however, completely disappears.  It appears in hidden form through the theory of latent variables.


\subsection{Revisiting: latent variables}

In this subsection, we establish that a Willems system $(\U,\B)$ is a manifest system of $(\U\times\U_l,\B_f)$ with latent variables if, and only if, $\B \ra \U$ is a subsystem of $\B_f \ra \U\times\U_l$.
\begin{Pro}
  Let $s : B \ra U\times U_l$ be a system, and consider a surjective map $U \xra{\pi} U'$, then (up to isomorphism) there is a unique set $B'$ and a unique surjective map $p$ such that:
    \begin{equation*}
  \begin{CD}
    B  @>s>> U\\
    @VVpV         @VV{\pi}V\\
    B'  @>>>   U'
  \end{CD}
  \end{equation*}
is a morphism of systems, i.e. is commutative with $B'\ra U'$ injective. Furthermore, $B'$ is isomorphic to $\im(\pi s)$, the image set of $\pi s$.
\end{Pro}

\begin{proof}
 Every set map $f: A \ra B$ has a unique factorization $f = is$ where $s$ is surjective and $i$ is injective. Indeed, let $is$ and $i's'$ be two factorizations, then $is(A)$ and $i's'(A)$ have the same cardinality. Since $i$ and $i'$ are injective, then $s(A)$ and $s'(A)$ have the same cardinality, and so are isomorphic.  As sets, they are isomorphic to $\im(f)$.
 \end{proof}
 
Specifying latent variables amounts to specifying a projection map $\U$ onto $\U'$.  In particular, the surjective map $\U \times \U_l \ra \U$ that forgets the components of $\U_l$ and maps the components in $\U$ identically onto $\U$ automatically identifies the set $\U_l$ as the universum of latent variables.  This surjective map induces a unique subsystem by projecting the full behavior onto the manifest behavior.

\begin{Cor}
  Let $s:B \ra U$ be a system, and $\pi: U \times U_l \ra U$ be the projection onto the first coordinate $(u,u_l) \mapsto u$, then $\im(\pi s) = \{ b \in U : (b,l)\in B \text{ for some }l\}$. \hfill \qed
\end{Cor}

Every Willems system with latent variables uniquely defines its manifest system as a subsystem.  However some subsystems cannot be realized as a manifest system of some Willems system of latent variables. The notion of subsystem then properly subsumes the notion of latent variables.

\begin{Exa}
 Referring back to the two-port blackbox abstraction of (P), recall that $\U_{two-port}$ is the free $\bR$-vector space with basis $\{v_{e'},v_{f'},v_{i'},v_{j'},i_{e'},i_{i'},i_{f'},i_{j'}\}$.  Define the map $\phi: \U_P \ra \U_{two-port}$ to be the projection that sends $v_{e}$,$v_{f}$,$v_{i}$,$v_{j}$,$i_{eg}$,$i_{gi}$, $i_{fh}$,$i_{hj}$ in the basis of $\U_P$ respectively to $v_{e'}$,$v_{f'}$,$v_{i'}$,$v_{j'}$,$i_{e'}$,$i_{i'}$,$i_{f'}$,$i_{j'}$, and everything else remaining in the basis of $\U_P$ to $0$.  The unique subsystem $\bphi$ induced by a $\phi_U$ component equal to $\phi$ has the system $\B_{two-port} \ra \U_{two-port}$ as a codomain.  It thus defines $(\B_{2-port} \ra \U_{2-port},\bphi)$ as a subsystem of $\B_P \ra \U_P$.
\end{Exa}

Subsystems will have their use in interconnection of systems.  We generally think of the manifest variables as the variables that we wish to model, or rather that are of interest.  We may then think of them as being the variables of interest when it comes to interconnecting two systems.  More precisely, we can think of them as being the variables that two systems will share.  As a pick of variable (i.e., a projection) directly induces a subsystem, we may think of interconnection as sharing a common subsystem. The problem is that two systems may not share a common non-trivial subsystem. There are many pairs of systems $s$ and $s'$, where if $\phi: s \dotarrow s_{sub}$ and $\phi: s' \dotarrow s_{sub}$ are morphisms such that $(s_{sub},\phi)$ and $(s_{sub},\phi')$ are subsystems of $s$ and $s'$, then $s_{sub}$ is the trivial identity map over the set with one element.  We can however relax the notion of subsystem to that of a quasi-subsystem.  Then every pair of systems would share a non-trivial quasi-subsystem in common.

\begin{Def}[Quasi-Subsystem]
Let ${\bphi} : s \dotarrow s_{qsub}$ be a morphism of systems.  The pair $(s_{qsub},\bphi)$ is said to be a quasi-subsystem of $s$, if the second component $\phi_U$ of $\bphi$ is surjective. We may refer to $s_{qsub}$ as the quasi-subsystem if ${\bf\phi}$ is clear from the context.
\end{Def}

\subsection{Revisiting: interconnection and variable sharing}

The behavioral approach encourages that systems be made from smaller pieces by identifying variables together.  We introduce a construct, termed \emph{category}, to aid in capturing the pattern of this idea. Our systems along with their morphisms will form a category.  The reason for introducing categories is that interconnection and variables sharing amounts only to an instantiation of a general construction known as \emph{pullback}.

\subsubsection{Interlude on categories}
A category can be simply viewed as a directed multi-graph, where the arcs can be composed associatively, and every node has a self-arc that produces no effect when composed.

\begin{Def}[Category]
  A category $\mathcal{\C}$ consists of:
\begin{itemize}\setlength\itemsep{0em}
  \item[i.] A class $\Obj(\mathcal{C})$ of objects.
  \item[ii.] A set $\Hom_{\mathcal{C}}(A,B)$ of morphisms for every ordered pair $(A,B)$ of objects.  A morphism $f$ in $\Hom_{\C}(A,B)$ is denoted by $f: A \ra B$. If $f: A \ra B$ is a morphism, we refer to $A$ and $B$ as the domain and codomain of $f$.
  \item[iii.] A composition map $\Hom(A,B)\times \Hom(B,C) \ra \Hom(A,C)$ for every ordered triple $(A,B,C)$ of objects. The composition of two morphisms $f:A \ra B$ and $g: B \ra C$ is denoted by either $g\circ f$ or $gf$.
    \item[vi.] An identity morphism $\id_A \in \Hom(A,A)$ for every object $A$.
\end{itemize}
This data is subject to two axioms:
  \begin{itemize} \setlength\itemsep{0em}
    \item[A.1.] For every $f:A \ra B$, $g: B \ra C$ and $h:C \ra D$, we have $(f\circ g) \circ h = f\circ(g\circ h)$.
    \item[A.2.] For every $f: A \ra B$, we have $f\circ \id_A = \id_B\circ f = f$.
  \end{itemize}
\end{Def}

The primordial example of a category is the category {\bf Set} where the objects are sets, and morphisms are set functions. Other typical examples may include the category of vector spaces (over a fixed field) with linear maps, the category of topological spaces with continuous maps, the category of groups with group homomorphisms.  On a different end, every partially-ordered set forms a category by declaring the elements of the set as the objects, and having $\Hom(A,B)$ contain exactly one morphism (and none otherwise) if, and only if, $A \leq B$ in the partial order.  Composition, in such a case, reflects the transitive property.  The presence of the identity arrow, reflects reflexivity.

\begin{Pro}
 The systems (i.e., the injective maps $B \ra U$) along with their morphisms (defined in subsection \ref{morphism}) form a category.
\end{Pro}

\begin{proof}
  The conditions were already verified in section \ref{morphism}.
\end{proof}

  We denote by {\bf System} the category of systems whose objects are injective maps $B \ra U$ and morphims $\phi$ are defined in subsection \ref{morphism}.

\subsubsection{Back to interconnection}

We introduce a universal construction in general categories termed pullback (or fibered-product).  Once instantiated to our category {\bf System}, it directly recovers the notion of interconnection and variable sharing.

\begin{Def}[Pullback]
  Let $\C$ be a category, and let $f_1:A_1 \ra B$ and $f_2:A_2 \ra B$ be two morphisms in $\C$ with the same codomain.  The pullback of $(f_1,f_2)$ consists of a triple $(K, p_1,p_2)$ where $K$ is an object of $\C$ and $p_1: K \ra A_1$ and $p_2: K \ra A_2$ are morphisms such that $f_1p_1 = f_2p_2$ satisfying the following universal property: for every other triple $(H,q_1,q_2)$ such that $q_1f_1=q_2f_2$, there is a unique morphism $h: H \ra K$ for which the following diagram commutes:

\begin{equation*}
  \begin{tikzcd}
    H
    \arrow[drr, bend left, "q_1"]
    \arrow[ddr, bend right, "q_2"]
    \arrow[dr, dotted, "h" description] & & \\
    & K \arrow[r, "p_1"] \arrow[d, "p_2"]
    & A_1 \arrow[d, "f_1"] \\
    & A_2 \arrow[r, "f_2"] & B
  \end{tikzcd}.
\end{equation*}
  We refer to $K$ as the \emph{object of the pullback}, denoted by $A_1 \times_B A_2$.
\end{Def}

If $f_1: A_1\ra B$ and $f_2 : A_2 \ra B$ are set maps in ${\bf Set}$, then the pullback $(K,p_1,p_2)$ consists of $K = \{(a_1,a_2) \in A_1 \times A_2 : f_1a_1 = f_2a_2\}$. We consider only the object of the pullback in this paper, and generally forget about the remaining maps.  The definition of a pullback instantiates to {\bf System} as:

\begin{Pro}
  Let $s: B \ra U$, $s': B'\ra U'$ and $s_c :B_c \ra U_c$ be systems, and let $\bphi = (\phi_B,\phi_U): s \dotarrow s_c$ and $\bphi = (\phi'_B,\phi'_U): s' \dotarrow s_c$ be morphisms of systems.  Then $s \times_{s_c} s'$ is the system $B^* \ra U^*$ where $B^* = \{ (b,b') \in B\times B' : \phi_Bb = \phi'_Bb'\}$, $U^* = \{ (u,u') \in U\times U' : \phi_Uu = \phi'_Uu'\}$ and the set map $B^* \ra U^*$ is the restriction of the product map $s \times s': B \times B' \ra U \times U'$ to the domain $B^*$.
\end{Pro}

\begin{proof}
  Let $s^*$ be $B^* \ra U^*$, and $\bphi^* : s^* \dotarrow s \times s'$ be the canonical controlled-system.  The pullback is $(s^*, \pi_s\bphi^*, \pi_{s'}\bphi^*)$ where $\pi_s$ and $\pi_{s'}$ are the projections from $s \times s'$ to $s$ and $s'$ respectively. Commutativity of diagrams and the universal property can be easily checked, and follow from the case of ${\bf Set}$.
\end{proof}

We return to variable sharing. Let $s:B \ra U \times U_c$ and $s': B' \ra U' \times U_c$ be two systems, and suppose we want to share the variables given by the universum $U_c$.  We then have projections $p: U \times U_c \ra U_c$ and $p': U' \times U_c \ra U_c$.  We pick $(s_c,\bphi)$ and $(s_c,\bphi')$ to be two arbitrary quasi-subsystems of $s$ and $s'$ respectively, where $\phi_U$ and $\phi'_U$ are $p$ and $p'$ respectively.  Such quasi-subsystems always exist. For instance, we may trivially pick $s_c = \id_{U_c}$, $\bphi = (ps,p)$ and $\bphi' = (p's',p')$.  The interconnected system obtained by sharing the variables in $U_c$ is then given by the pullback of $\bphi$ and $\bphi'$.  Any pair of \emph{common} quasi-subsystems, whose $U$ components are $p$ and $p'$, yields the same interconnected system.

\begin{Cor}
 Let $(s_c,\bphi)$ and $(s_c,\bphi')$ be two quasi-subsystems of $s:B \ra U \times U_c$ and $s': B' \ra U' \times U_c$ respectively.  Suppose further that $s_c$ is a map $B_c \ra U_c$ and that $\phi_U$ and $\phi'_U$ are the projections $p: U \times U_c \ra U_c$ and $p': U' \times U_c \ra U_c$.  Then, the object of the pullback of $\bphi$ and $\bphi'$ represents the system $B \times_{B_c} B' \ra (U\times U_c) \times_{U_c} (U'\times U_c)$ where the universum corresponds to $U \times U_c \times U'$ and the behavior to $B \times U' \cap U \times B'$. \hfill \qed
\end{Cor}

\begin{Exa}
 We define $s_c$ to be $\id: \U_c \ra \U_c$ where $\U_c$ is the free $\bR$-vector space with basis $\{v_{c=e},v_{d=f}\}$. Let $(s_c,\bphi)$ and $(s_c,\bphi')$ be two quasi-subsystems of $s:\B_S \ra \U_S$ and $s': \B_P \ra \U_P$ respectively. The component $\phi_U$ is defined to send $v_c,v_d$ in the basis of $\U_S$ respectively to $v_{c=e},v_{d=f}$, and everything else remaining in the basis of $\U_S$ to $0$.  Similarly, the component $\phi'_U$ is defined to send $v_{e},v_{f}$ in the basis of $\U_P$ respectively to $v_{c=e},v_{d=f}$, and everything else remaining in the basis of $\U_P$ to $0$.  The system corresponding to the circuit:
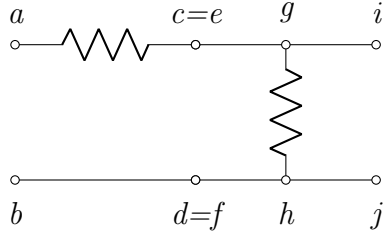
\begin{center}
\begin{circuitikz}[scale=0.6]\draw
  (0,3) node[label=a]{} to[R, o-o] (4,3)
  (0,0) node[label=below:b]{} to [short, o-o] (4,0)
  (4,0) node[label=below:{d=f}]{} to [open, o-o] (4,3) node[label={c=e}]{} -- (8,3) node[label=i]{} to[open, o-o] (8,0) node[label=below:j]{} -- (4,0)
  (6,3) node[label=g]{} to[R, o-o] (6,0) node[label=below:h]{};
\end{circuitikz}
\end{center}
is then the object of the pullback of $\bphi$ and $\bphi'$. 
\end{Exa}

If $\{*\}$ denotes \emph{the} set with one element, then $1 = \{*\} \ra \{*\}$ is a system.  Furthermore, for every system $s$, the unique morphism $s \dotarrow 1$ is a subsystem. In case we pullback from $s \dotarrow 1$ and $s' \dotarrow 1$, we are indicating that $s$ and $s'$ are not sharing any variables. The system we recover from their interaction is then simply the two systems independently put together.  The system we recover from the pullback is just the product system $s \times s'$ corresponding to the product of the universum and the product of the behavior.

 {\bf Important remark}. Let $\bphi : s \dotarrow s_c$ and $\bphi : s' \dotarrow s_c$ be two morphisms, non-necessarily quasi-subsystems, and denote by $(K,\pi,\pi')$ the pullback of $\bphi$ and $\bphi'$.  In that case, $\pi$ and $\pi'$ are morphisms of systems, and we may form the morphism $\pi \times \pi': K \dotarrow s \times s'$. The pair $(K,\pi \times \pi')$ will always be a controlled-system.  Thus variable sharing then consists of pulling back along a common quasi-subsystem to yield a controlled-system on the product of the systems, i.e. on the separate systems simply put next to each other.

 {\bf Limits}. In case we wish to interconnect more than two systems at a time, we can generalize pullbacks to (projective) limit. Such a generalization will not be studied in this paper.

 \subsection{Revisiting: behavioral equations}

 Behavioral equations consists of two morphisms with the same domain and codomain.  The behavior they represent then corresponds to their equalizer. 

\begin{Def}[Equalizer]
  Let $\C$ be a category, and $f, g : A \ra B$ two morphisms. An equalizer of $f$ and $g$ is a pair $(E,e:E \ra A)$ such that $fe=ge$ and for every map $h: H \ra A$ such that $fh=gh$ there is a unique map $u : H \ra E$ such that $h = eu$.
\end{Def}
Every pair of maps between two sets have an equalizer.  It is the subset of $A$ that is sent to the same image in $B$ by both maps.  As such, the equalizer will always be injective.  We can immediately recognize that the behavior of the system is the equalizer to the behavioral equations.  Therefore:

\begin{Pro}
  The triple $(\U,\mathbb{E},f_1,f_2)$ is a \emph{behavioral equation representation} of $(\U,\B)$ if and only if $(\B, \B\ra\U)$ is the equalizer of $f_1$ and $f_2$. \hfill\qed
\end{Pro}

However, more structure can be harvested. Let us define a category {\bf Equation} of behavioral equation representations.  The objects of {\bf Equation} are pairs $(f_1,f_2)$ of set maps $U \ra E$ for every pair of sets $U$ and $E$. A morphism $\bpsi$ from $(f_1,f_2):U \ra E$ to $(g_1,g_2):U' \ra E'$ is also a pair of maps $(\psi_U,\psi_E)$ with the commutativity properties $\psi_E f_1 =  g_1\psi_U$ and $\psi_E f_2 =  g_2\psi_U$.

One should think of {\bf Equation} as a category of syntax, i.e. of descriptions of systems.  This parallels {\bf System} which acts as the category of semantics.  Pullbacks in the category {\bf Equation} instantiates to constructing our systems syntactically.  Indeed, pullbacks in {\bf Equation} consist of stacking the equations together while making sure that common variables are well identified and taken care of.  We will not flush out the details of this. The two categories are related through the equalizer rule.  We make this precise by introducing the notion of a functor.

\subsubsection{Interlude on functors}

A functor is a mapping that preserves the structure of a category. If categories are viewed as directed multi-graphs, then functors are graph homomorphisms that preserve composition and identity.

\begin{Def}
 A functor $F: \mathcal{C} \ra \mathcal{D}$ from a category $\mathcal{C}$ to a category $\mathcal{D}$ is a rule that assigns an object $F(C)$ of $\mathcal{D}$ to every object $C$ of $\mathcal{C}$, and a morphism $F(f):F(C_1) \ra F(C_2)$ in $\mathcal{D}$ to every morphsim $f:C_1 \ra C_2$ in $\mathcal{C}$.  We require a functor to preserve the identity morphisms, i.e.,  $F(\id_C) = \id_{F(C)}$ and composition, i.e., $F(f\circ g) = F(f) \circ F(g)$.
\end{Def}

Consider, as an example, the power-set functor that sends a set $A$ to $2^A$ its set of subsets. It also sends functions $A \ra B$ to functions $2^A \ra 2^B$ by mapping a subset of $A$ to its image set. It forms a functor from {\bf Set} to {\bf Set}.  On a different end, every order-preserving map between partially-ordered sets is a functor on the induced categories.

\subsubsection{Back to the behavioral equations}

Let {\bf Arr{-}Eq} be the map that sends an object in $\Obj(${\bf Equation}$)$ to the morphism component in its equalizer in $\Obj(${\bf System}$)$.

\begin{Pro}
 The map {\bf Arr{-}Eq} lifts to a functor ${\rm\bf Equation} \ra {\rm\bf System}$.
\end{Pro}
\begin{proof} If $(\psi_U,\psi_E):(f_1,f_2) \dotarrow (g_1,g_2)$ is a morphism in ${\rm\bf Equation}$, we let ${\rm\bf Arr{-}Eq}(\psi_U,\psi_E)$ be the unique morphism $\bphi : {\rm\bf Arr{-}Eq}(f_1,f_2) \dotarrow {\rm\bf Arr{-}Eq}(g_1,g_2)$ in ${\rm\bf System}$ such that $\phi_U = \psi_U$.
\end{proof}

The crucial property that connects syntax and semantics well is that ${\rm\bf Arr{-}Eq}$ preserves pullbacks. The semantics reflect interconnection in the syntax. 

\begin{Pro}
 Let $\bpsi : e \dotarrow e_c$ and $\bpsi' : e' \dotarrow e_c$ be morphisms in {\bf Equation}, then $\text{\bf Arr{-}Eq}$ takes the object of the pullback of $(\bpsi,\bpsi')$ to the object of the pullback of $(\text{\bf Arr{-}Eq}\bpsi,\text{\bf Arr{-}Eq}\bpsi')$. 
\end{Pro}
\begin{proof} Let $(e^*,\bpsi^*,\bpsi'^*)$ be the pullback of $(\bpsi,\bpsi')$.  As equalizers preserve limits, and in particular pullbacks, the domain $B^*$ of ${\rm\bf Arr{-}Eq}e^*$ coincides with the domain of the object of the pullback of $({\rm\bf Arr{-}Eq}\bpsi,{\rm\bf Arr{-}Eq}\bpsi')$. The codomains trivially coincide, and are denoted by $U^*$. Finally there is a unique set map $B^* \ra U^*$ ensuring that the emerging diagram commutes.
\end{proof}
We first illustrate with a clear example.

\begin{Exa}
  We consider two resistors (R) and (R'), and define a syntactical interconnection, on the level of behavioral equation representations.
\begin{center}
  \begin{circuitikz}[scale=0.6]\draw
 (-1,0) node[]{(R)}
 (7,0) node[]{(R')}
 (0,0) node[label=a]{} to[R, l=R, o-o] (4,0) node[label=b]{}
 (8,0) node[label=c]{} to[R, l=R', o-o] (12,0) node[label=d]{};
\end{circuitikz}
\end{center}
The universa $\U_{R}$ and $\U_{R'}$ are the free $\bR$-vector spaces with basis $\{v_a,v_b,i_{ab}\}$ and $\{v_c,v_d,i_{cd}\}$, respectively. A behavioral equation representation of (R) consists of a pair $(f_R,0)$ where $f_R : \U_R \ra \bR$ sends $v_a,v_b,i_{ab}$ to $-1,1,R$ respectively.  Likewise, a behavioral equation representation of (R') consists of a pair $(f_{R'},0)$ where $f_{R'} : \U_{R'} \ra \bR$ sends $v_c,v_d,i_{cd}$ to $-1,1,R'$ respectively.  Let $\bR^2$ be the free $\bR$-vector space with basis $\{v,i\}$ and let $0_2 : \bR^2 \ra 0$ be the zero map.  We can uniquely define morphisms $\bpsi : (f_{R},0) \dotarrow (0_2,0_2)$  and $\bpsi' : (f_{R'},0) \dotarrow (0_2,0_2)$ such that $\psi_{U}:\U_R \ra \bR^2$ sends $v_a,v_b,i_{ab}$ to $0,v,i$ and $\psi'_{U} : \U_{R'} \ra \bR^2$ sends $v_c,v_d,i_{ab}$ to $v,0,i$.  Pulling back along $\bpsi$ and $\bpsi'$ yields as object the pair $(f,0): \bR^4 \ra \bR^2$ which is a behavioral equation representation for the series circuit once (R) and (R') are interconnected at $b$ and $c$.
\end{Exa}

We may now link to the running example:

\begin{Exa}
  Let us suppose that $\bR^2$ has basis $\{e_1,e_2\}$. A behavioral equation representation of (S) can be identified with a linear map $f_S:\U_S \ra \bR^2$, that sends $v_a,v_b,v_c,v_d,i_{ac},i_{bd}$ to $-e_1,-e_2,e_1,e_2,-Re_1,0$.  Clearly, the kernel of $f_S$ gives $\B_S$. Let $f_c$ be the identity map $\id_2: \bR^2 \ra \bR^2$.  We may now construct the unique morphism $\bpsi : (f_S,0) \ra (\id_2,0)$ such that $\psi_E$ is the identity map. A behavioral representation of (P) can be identified with a linear map $f_P: \U_P \ra \bR^9$, and a canonical morphism $\bpsi' : (f_P,0) \ra (\id_2,0)$ may be set up.  The object of the pullback along $\bpsi$ and $\bpsi'$ yields a behavioral equation representation of the interconnection of (S) and (T) by identifying $c$ with $e$ and $d$ with $f$.
\end{Exa}

{\bf Important Remark:}  It is very desirable that each system possesses a simplest behavioral equation representation of it.  The presented syntax category along with the {\bf Arr{-}Eq} functor will not allow such a thing. Every system will have many behavioral equations attached to it, but no \emph{universal} one.  %
We can however get a much better category and functor by working with generalized systems, where we drop the injectivity condition on the maps $B \ra U$.  Each generalized system then has a \emph{simplest} description to it. 
The new category preserves all the features described in this section, but consists of a more involved construction. To preserve flow of the paper, we present the desired category in an appendix section on generalized systems.

\section{The big picture}

We constructed two categories.  A category of syntax {\bf Equation}, reflecting syntactical descriptions of systems through behavioral equations, and a category of semantics {\bf System}, reflecting the objects behind the equations, the behaviors.  Those two categories are linked by an interpretation functor ${\rm\bf Arr{-}Eq}$ that sends each description to the underlying system it described. Interconnection of systems in {\bf System} consist of taking pullbacks.  We uncovered a notion of (quasi-)subsystem, and a notion of controlled-system.  Indeed, variable sharing consists of pulling back along a shared quasi-subsystem to get a controlled-system.  Syntactical constructions in {\bf Equation} also consist of taking pullbacks. Finally, the interpretation functor ${\rm\bf Arr{-}Eq}$ preserves pullbacks, thus translating syntactical interconnection to semantical interconnection.

Throughout the development, sub-systems were identified with surjective maps, and controlled-systems with injective maps.  Such an association might seem unnatural as we already described in subsection \ref{subcon}.  Indeed, the prefix sub of subsystem typically alludes to a part of bigger system, namely a possibility of embedding. It then seems counter-intuitive that surjective maps rather than injective maps are involved.  Also, controlled systems lead us to think of identifying parts of the system as a means of control.  This identification is intuively perceived through surjective maps.  To recover a more intuitive association, we then only need to flip the direction of the morphisms.  Such a flip is discussed in the next sub-section.

\subsection{The opposite approach}

Rather than thinking of a Willems system as a pair $(\U,\B)$ with $\B \subseteq \U$, we will think of it as a pair $(2^\U,2^\B)$, where $2^S$ denotes the set of subsets of a set $S$. Then, from the relative point of view, a system is no longer an injective map $B \ra U$, but a surjective map $2^U \ra 2^B$ that sends $S\in 2^U$ to $S\cap B$.  Such a consideration might seem absurd, but it is equally valid and may appear to be more natural.  Subsystems are now associated with injective morphisms, and controlled-systems are now associated with surjective morphisms. Overall, this transformation only reverses the \emph{direction} of the morphisms involved.

Let us define {\bf Bool} to be the category whose objects are complete atomic boolean lattices, and whose morphisms are the lattice homomorphisms. Every complete atomic boolean lattice is isomorphic to $2^S$ for some set $S$.  We may then just consider the objects as lattices of the form $2^S$ for some set $S$.  The categories {\bf Set} and {\bf Bool} may be related through two functors.
\begin{itemize}

\item[i.] {\bf From Set to Bool}. If $f:S \ra T$ is a set map,  then $f^{-1}: 2^T \ra 2^S$ is a lattice homomorphism. We denote by $F : {\rm\bf Set} \ra {\rm\bf Bool}$ the functor that sends $S$ to $2^S$ and $f:S \ra T$ to $f^{-1}: 2^T \ra 2^S$.

\item[ii.] {\bf From Bool to Set}. If $\phi: 2^T \ra 2^S$ is a lattice homomorphism, then there exists a unique map $G\phi : S \ra T$ such that $(G\phi)^{-1} = \phi$. Indeed, as $\phi \{t\} \cap \phi \{t'\} = \emptyset$ for every $t\neq t'$ in $T$, every $s$ belongs to a set $\phi\{t\}$ for a unique $t$. If the map $G\phi$ sends every $s \in \phi\{t\}$ to $t$, then $(G\phi)^{-1} = \phi$.  We denote by $G : {\rm\bf Bool} \ra {\rm\bf Set}$ the functor that sends $2^T$ to $T$ and $\phi$ to $G\phi$ as just described.
\end{itemize}
We make the relation between {\bf Bool} and {\bf Set} precise by defining the notion of an opposite category.
\begin{Def}[Opposite Categories]
 The opposite category $\mathcal{C}^{op}$ of $\mathcal{C}$ is the category containing the objects of $C$, and a morphism $f^{op}: B \ra A$ for every $f : A \ra B$ in $C$. If $f^{op} : B \ra A$ and $g^{op} : C \ra B$ are in $\mathcal{C}^{op}$, then $f^{op}\circ g^{op} = (g \circ f)^{op}$.
\end{Def}

The notion of opposite categories, in general, allows us to dualize both definitions and results.

\begin{Def}[Pushouts]
The pushout of $f: A \ra B$ and $g: A \ra C$ in $\C$ is the pullback of $f^{op}$ and $g^{op}$ in $\mathcal{C}^{op}$.
\end{Def}

\begin{Def}[Coequalizers]
  The coequalizer of $f_1: A \ra B$ and $f_2: A \ra B$ in $\C$ is the equalizer of $f^{op}$ and $g^{op}$ in $\mathcal{C}^{op}$.
\end{Def}  

The key fact is that {\bf Bool} is \emph{equivalent} {\bf Set}$^{op}$. Such a fact lifts to our category of systems.  We shall not formalize the term \emph{equivalent}, but rather describe what it entails for us. For all purposes, working in {\bf Bool} is the same as working in {\bf Set}$^{op}$.  The functors $F$ and $G$ give us the mean to flip the arrows.  The equivalence manifests itself first in:

\begin{Pro}\label{injectivesurjective}
  A map $f$ is injective (resp. surjective) in {\bf Set} if, and only if, $Ff = f^{-1}$ is surjective (resp. injective) in {\bf Bool}.
\end{Pro}

\begin{proof} The result follows from the definition of $F$.
\end{proof}

We can then define a category {\bf Bool-System} to have, as objects, surjective boolean lattice homomorphisms $2^U \ra 2^B$.  A morphism $\bpsi$ in {\bf Bool-System} is then a pair of homomorphisms $(\psi_U, \psi_B)$ such that the diagram:
  \begin{equation*}
  \begin{CD}
    2^U  @>h>> 2^B\\
    @VV\psi_UV         @VV{\psi_B}V\\
    2^{U'}  @>{h'}>>   2^{B'}
  \end{CD}
  \end{equation*}
  commutes, i.e., such that $h'\psi_U = \psi_Bh$.

  Building on proposition \ref{injectivesurjective}, one can check that $h$ (resp. $Fs$) is an object of {\bf Bool-System} if, and only if, $Gh$ (resp. $s$) is an object of {\bf System}. A morphsim $\bpsi = (\psi_U,\psi_B)$ is a morphim of {\bf Bool-System} if, and only if,  $G\bpsi = (G\psi_B,G\psi_U)$ is a morphism of {\bf System}. Most importantly:
  
\begin{Pro}
 The bool-system $(h,i,i')$ is the pushout of $\bpsi$ and $\bpsi'$ if, and only if, $(Gh,Gi,Gi')$ is the pullback of $G\bpsi$ and $G\bpsi'$.
\end{Pro}

\begin{proof} As ${\rm\bf Bool}$ is equivalent to ${\rm\bf Set}^{op}$, the category ${\rm\bf Bool-System}$ is equivalent to the category ${\rm\bf System}^{op}$.  The result follows by definition.
\end{proof}

We are only translating diagrams in {\bf System} to diagrams in {\bf Bool-System} without losing any property at all.

A subsystem in {\bf Bool-System} is then a morphism where both components are injective. A controlled-system in {\bf Bool-System} is then a morphism where both components are surjective.  A part of a universum is also now described as an embedding (i.e., injection) rather than a projection (i.e., surjection). The intuition is thus remedied.  One way to construct topological spaces consists of gluing simple ones together by identifying subspaces. Such a construction is an instantiation of the pushout.  Similarly, we may construct complicated systems from simple ones, by identifying subsystems together along a part of the universum, to get a controlled-system.

\subsubsection{Discussion and remarks} A system restricts (or excludes) some outcomes from the universum to yield the behavior.  The map $2^\bU \ra 2^\B$ can then be seen as the object that encodes the "restrictive capability" of the system.  It acts similarly to the equations defining a system, but it is actually independent of the particular (behavioral) equation representation.  A rough analogy to this duality would be the duality between varieties and polynomial ideals.  The variety is a geometric object encoding the solution, and the polynomial ideal (or the affine coordinate ring) is an algebraic object encoding the equations.

Furthermore, an epi from $2^\bU$ to $2^\B$ induces a kernel (dual to a closure) operator on the subsets of $\bU$.  For every subset $V$ of outcomes in $\bU$, a Willems' system defines which outcomes of $V$ are possible.  The behavior of the system can also be recovered from $B=\Hom(2^B,2^{*})$. A subsystem is then an objects that is a sub-restriction, whereas a controlled system is an object that is a over-restriction. Again, the object $2^\bU \ra 2^\B$ is independent of the particular \emph{syntax}.

 We believe this line of thought for the behavioral approach has been well treated using $D$-module theory through algebraic analysis.  The link will not be investigated in this paper, we instead refer the reader to \cite{QUA2010} and \cite{OBE1990} for an initial thread.

\subsection{How essential is the category {\bf System}?}

The systems in {\bf System} are defined as injective maps $B \ra U$.  We may just keep the domain of the map, thus keeping only the behavior of the systems.  We definitely lose information. Indeed if $B \subseteq U$ and $B' \subseteq U$ are distinct subsets of $U$, we will have no way of distinguishing them in case they have the same cardinality if we forget the map.  Nevertheless, forgetting that information preserves constructions of systems, i.e. pullbacks, and may be a valid definition of systems for some instances.

This idea however suggests adding more information, rather than forgetting some, for additional effects.  Given that universa are relative, we may decide to make some elements of the universa distinguished.  We may further consider only sets with additional structure, e.g. vector spaces or in general modules over rings. Those objects inherit the same notion of interconnection, and a development may be carried out along the same line as that carried out here. Thus, any \emph{content} of the category is valid, as long as we can have a good interpretation for it as a system. It will nevertheless make a technical difference.  However, the intrinsic pattern of interconnection and variable sharing, as well as the pattern of the other ideas we have described, will remain unchanged.

{\bf Remark:} Categories other than {\bf Equation} may also act as a category of syntax.  Following the lines of the running example, we should expect circuit diagrams, if formalized properly, to act as an alternative category of syntax, when our systems are restricted to underly resistive circuits.

\section{The concluding picture}
We abstract away to the following scheme:
\begin{equation*}
  \begin{tikzcd}
    {\rm\bf System\ Syntax} \arrow[rr,"{\rm\bf Interp}"] & & {\rm\bf System\ Semantics} 
  \end{tikzcd}
\end{equation*}
Constructions both in the syntax and semantics category consist of taking pushouts, and the functor {\bf Interp} preserves (or commutes with) pushouts. Pushouts may replace pullbacks, without loss of generality, by working with the opposite categories. Pushouts may be generally replaced by colimits (or inductive limits), and then {\bf Interp} would then be required to preserve colimits.  In general categories, the role played by injective (resp. surjective) maps in {\bf Set} will be played by monic (resp. epic) morphism. Thus subsystems would correspond to monos, and controlled-systems to epis.  A mono $A \ra B$ defines $A$ as a subobject of $B$, while an epi $B \ra C$ defines $C$ as a quotient-object of $B$. In general, additional properties would be required by the monos and epis, e.g., regularity, to suit our needs. We refer to \cite{MAC1997} for further details on functorial matters.

The functor {\bf Interp} is envisioned to have more properties than what will be explicitly mentioned in this paper. For instance, in our setting, every system admits at least one behavioral equation representation, or description.  It is very pleasing to have every system have a \emph{simplest} description.  The syntax category thus provided does not afford that.  We may construct one that does by going to generalized systems.  When such a situation happens, {\bf Interp} is said to have an adjoint.

\subsection{In the case of partially ordered sets}

Every partially ordered set (poset) forms a category by declaring the elements of the set as the objects and having $\Hom(A,B)$ contain one morphism if, and only if, $A \leq B$ in the poset.  A functor between two categories induced by posets is then only an order-preserving map.

We consider, in this subsection, only posets where every pair of elements admits a least upper bound.  The operation of taking least upper bounds can be declared as a binary operation $\vee$, termed \emph{join}.  The algebras obtained are then termed join-semilattices. Taking a pushout along two morphims (when posets are viewed as categories) consists of taking the join (in the join-semilattice) of the their codomains.

We may then consider the categories {\rm\bf System Syntax} and {\rm\bf System Semantics} to be induced by posets.  The functor {\bf Interp} is then an order-preserving map that commutes with the join operation.  If $s \leq s'$, then $s$ is a subsystem of $s'$ and $s'$ is a controlled-system from $s$. We can acquire a physical interpretation if we consider {\rm\bf System Semantics} to be the lattice of behaviors over a fixed universum.

\subsection{Intended application}

We return to our main concern of uncovering phenomena that emerge from the interaction of systems.  A theory of interconnection cannot be enough to account for interaction-related effects.  Interconnecting two systems can only yield an interconnected systems.  Such effects may only emerge once we decide to focus on a feature or a property of the system, that we term \emph{phenome}.

We generally arrive at a phenome by forgetting, or concealing, information from the system. We may then think of the phenome as a simplified system.  Phenomes then live in a category and inherit a notion of interconnection through pushouts. The situation is summarized as:
\begin{equation*}
  \begin{tikzcd}
    {\rm\bf Phenome} & & {\rm\bf System} \arrow[ll,swap,"{\rm \bf Forget}"] 
  \end{tikzcd}
\end{equation*}
Whether or not new phenomena emerge upon the interaction of systems is now encoded in the functor {\bf Forget}.  New phenomes can emerge precisely when the {\bf Forget} functor does not commmute with pushouts, i.e., interconnections.  Indeed, the irrelevant things that we had forgotten actually come together and produce new observables.  If {\bf Forget} always commutes with pushouts, then the phenome of the interconnected system is simply the interconnection of the phenomes of the separate systems.
\begin{Exa}
  We reconsider the circuits (S) and (P) and augment them with \emph{external} terminals:
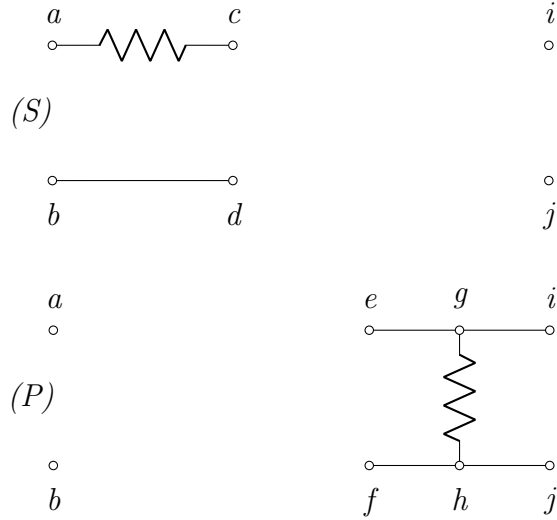
\begin{center}
\begin{circuitikz}[scale=0.6]\draw
  (-0.5,7.5) node[]{(S)}
  (0,9) node[label=a]{} to[R, o-o] (4,9) node[label=c]{}
  (0,6) node[label=below:b]{} to [short, o-o] (4,6) node[label=below:d]{}
  (11,9) node[label=i]{} to[open, o-o] (11,6) node[label=below:j]{};
\end{circuitikz}
\end{center}
\begin{center}
\begin{circuitikz}[scale=0.6]
  \draw
  (-0.5,1.5) node[]{(P)}
  (0,3) node[label=a]{} to[open, o-o] (0,0) node[label=below:b]{}
  (7,0) node[label=below:f]{} to [open, o-o] (7,3) node[label=e]{} -- (11,3) node[label=i]{} to[open, o-o] (11,0) node[label=below:j]{} -- (7,0)
  (9,3) node[label=g]{} to[R, o-o] (9,0) node[label=below:h]{};
\end{circuitikz}
\end{center}  
Each of the behavior underlying (S) and (P), when restricted to the variables $v_a$, $v_b$, $v_i$ and $v_j$, consists of a four dimensional $\bR$-vector space.
When we interconnect (S) and (P) by identifying $c$ to $e$ and $d$ to $f$, we obtain:
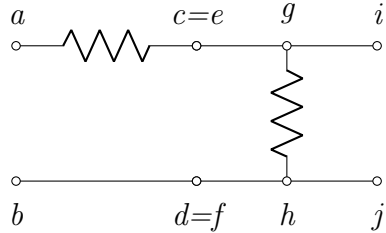
\begin{center}
\begin{circuitikz}[scale=0.6]\draw
  (0,3) node[label=a]{} to[R, o-o] (4,3) 
  (0,0) node[label=below:b]{} to [short, o-o] (4,0) 
  (4,0) node[label=below:{d=f}]{} to [open, o-o] (4,3) node[label={c=e}]{} -- (8,3) node[label=i]{} to[open, o-o] (8,0) node[label=below:j]{} -- (4,0)
  (6,3) node[label=g]{} to[R, o-o] (6,0) node[label=below:h]{};
\end{circuitikz}
\end{center}
The behavior of the underlying interconnected system, restricted again to the variables $v_a$, $v_b$, $v_i$ and $v_j$, now consists of a one dimensional $\bR$-vector space. The \emph{internal} mechanisms (declared \emph{internal} by focusing only on the terminals $a$, $b$, $i$ and $j$) in the circuits interact so as to produce new observables.
\end{Exa}
Using methods from homological algebra, we can then relate the phenome of the interconnected system to that of its subsystems, despite the presence of cascade-like effets.  We refer the reader to \cite{ADAM:Dissertation} for the details.

\bibliographystyle{alpha}
\bibliography{MyBiblio}

\section{Generalized systems}

Our systems were thus far injective set maps $B \ra U$.  We may drop the injectivity requirement and obtain the notion of a generalized system, which is only a set map.  We then establish a category {\bf Generalized-System} of generalized systems whose objects are set maps $C \ra U$, and morphisms $\bphi$ between $g:C \ra U$ and $g':C' \ra U'$ are pairs of maps $(\phi_C,\phi_U)$ such that the diagram:
  \begin{equation*}
  \begin{CD}
    C  @>g>> U\\
    @VV\phi_CV         @VV{\phi_U}V\\
    C'  @>g'>>   U'
  \end{CD}
  \end{equation*}
  commutes, i.e., such that $\phi_Ug = g'\phi_C$.

Given a generalized system $g : C \ra U$, we recover our regular interpretation of a system by reading its image, i.e., the injective map $g(C) \ra U$.  The rule mapping $g$ to $g(C) \ra U$ defines a functor from {\bf Generalized-System} to {\bf System}. Most of the notions defined earlier extend naturally to {\bf Generalized-System} while remaining intact on generalized-systems that are injective, i.e., the {\it regular} systems.

We now define our syntax category {\bf Generalized{-}Equation}.  The objects are pairs $\bphi^1, \bphi^2: g \ra g'$ of morphisms in {\bf Generalized-System} with the same domain and codomain. A morphism from $(\bphi^1, \bphi^2)$ to $(\bpsi^1,\bpsi^2)$ consists of four set maps $\tau_1, \tau_2, \tau_3, \tau_4$ such that the following diagram commutes for both values of $i$:
\begin{equation*}
  \begin{tikzcd}
    C \arrow[ddd,"\phi^i_C"] \arrow[dr,"\tau_1"] \arrow[rrr,"g"] & & & U \arrow[ddd, "\phi^i_U"] \arrow[dl, "\tau_2"]\\
    & D \arrow[r, "h"] \arrow[d, "\psi^i_D"] & V \arrow[d,"\psi^i_V"] & \\
    & D' \arrow[r, "{h'}"] & V' &     \\
    C' \arrow[rrr, "{g'}"] \arrow[ur,"\tau_3"] & & & U' \arrow[ul,"\tau_4"]
  \end{tikzcd}.
\end{equation*}

We may now define a functor:
\[
  {\rm\bf Obj{-}Eq} : {\rm\bf Generalized{-}Equation} \ra {\rm\bf Generalized{-}System}
  \]
  that sends a pair $\bphi^1, \bphi^2: g \ra g'$ to the object of its equalizer.

The objects of the category {\bf Equation} can be found in {\bf Generalized{-}Equation} as follows:
\begin{Pro}
  Let $(\bphi^1, \bphi^2)$ be an object in {\bf Generalized{-}Equation}, where $\bphi^1$ and $\bphi^2$ correspond to the following diagrams, respectively:
  \begin{equation*}
    \begin{tikzcd}
      U \arrow[r, "\id_U"] \arrow[d, "f_1"] & U \arrow[d, "*"] &&  U \arrow[r, "\id_U"] \arrow[d, "f_2"] & U \arrow[d, "*"] \\
      E \arrow[r, "*"] & \{*\} && E \arrow[r, "*"] & \{*\}
    \end{tikzcd}.
  \end{equation*}
Then {\bf Obj{-}Eq}$(\bphi^1, \bphi^2)$ is isomorphic to {\bf Arr{-}Eq}$(f_1,f_2)$. \hfill\qed
\end{Pro}

The functor {\rm \bf Obj{-}Eq} is right adjoint to the diagonal functor that sends the generalized-system $g$ to the generalized-equation $(\id_g ,\id_g)$. Every system $C \ra U$ then has a \emph{universal} description of the form:
  \begin{equation*}
    \begin{tikzcd}
      C \arrow[r, "g"] \arrow[d, "\id_C"] & U \arrow[d, "\id_U"] &&  C \arrow[r, "g"] \arrow[d, "\id_C"] & U \arrow[d, "\id_U"] \\
      C \arrow[r, "g"] & U && C \arrow[r, "g"] & U
    \end{tikzcd}.
  \end{equation*}
The functor {\rm \bf Obj{-}Eq} then preserves pullbacks, and more generally (projective) limits. 
 
\end{document}